\documentclass[letterpaper, 11pt]{article}

\usepackage{CJKutf8}

\usepackage[whole]{bxcjkjatype}

\usepackage{bookmark}
\usepackage{type1cm}
\usepackage{enumerate}
\usepackage{algorithm}
\usepackage{algorithmicx}
\usepackage[noend]{algpseudocode}
\usepackage[sort]{natbib}

\usepackage{url}
\usepackage{hyperref}
\usepackage{threeparttable}
\usepackage[margin=1in]{geometry}
\usepackage{amsmath,amssymb,amsfonts,setspace} 
\usepackage{amsthm} 
\usepackage{latexsym}
\usepackage{comment}
\usepackage{color}
\usepackage{xcolor}
\usepackage{bm}
\usepackage[normalem]{ulem}
\usepackage{mathtools}
\usepackage{stmaryrd} 

\usepackage[geometry]{ifsym}

\setcitestyle{numbers,square,comma}

\hypersetup{
setpagesize=false,
 bookmarksnumbered=true,%
 bookmarksopen=true,%
 colorlinks=true,%
 linkcolor=red,
 citecolor=blue,
}

\newtheoremstyle{mystyle}
    {}
    {}
    {\normalfont}
    {}
    {\bf}
    {}
    { }
    {}%

\theoremstyle{mystyle}
\newtheorem{theorem}{Theorem}
\newtheorem{lemma}{Lemma}

\newtheorem{corollary}{Corollary}
\newtheorem{definition}{Definition}

\newcommand{\mt}[1]{\mathit{#1}}

\floatstyle{ruled}
\newfloat{algorithmfig}{thp}{lof}
\floatname{algorithmfig}{Algorithm}

\newcommand{\ot}{\ensuremath{\leftarrow}}

\newcommand{\UlargerD}{\ensuremath{\mathrm{``Backward}>\mathrm{Forward''}}}
\newcommand{\UsmallerD}{\ensuremath{\mathrm{``Backward}<\mathrm{Forward''}}}
\newcommand{\Vcur}{\ensuremath{v_{\mathrm{cur}}}}

\newcommand{\Portin}{\ensuremath{\pi_{\mathrm{in}}}}

\newcommand{\Succ}{\ensuremath{\mathrm{fwd}}}
\newcommand{\Pred}{\ensuremath{\mathrm{back}}}
\newcommand{\True}{\ensuremath{\mathrm{True}}}
\newcommand{\False}{\ensuremath{\mathrm{False}}}

\newcommand{\Findfirst}{\textsf{FindFirst}}
\newcommand{\Mark}{\textsf{MarkPred}}
\newcommand{\Movetop}{\textsf{MoveSource}} 
\newcommand{\Modifymove}{\textsf{Modify\&Move}} 
\newcommand{\Moveonehopdown}{\textsf{MoveOneHopForward}}
\newcommand{\Movetarget}{\textsf{MoveTarget}}
\newcommand{\Modifysucc}{\textsf{ModifySuccessor}}
\newcommand{\Modifypred}{\textsf{ModifyPredecessor}}

\newcommand{\target}{\textsf{target}}
\newcommand{\inPath}{\textsf{inPath}}
\newcommand{\direction}{\textsf{direction}}
\newcommand{\Color}{\textsf{color}}

\newcommand{\mrm}[1]{\ensuremath{\mathrm{#1}}}
\newcommand{\msf}[1]{\ensuremath{\mathsf{#1}}}

\newcommand{\Sim}{\ensuremath{\mathit{Sim}}}

\newcommand{\TMSpace}[1]{\textsf{TMSpace}}

\newcommand{\SymLG}{\ensuremath{\mathsf{SymLG}}}
\newcommand{\Pagt}{\ensuremath{\mathsf{P_{AGT}}}}
\newcommand{\Poly}{\ensuremath{\mathrm{Poly}}}

\newcommand{\Nmem}{\ensuremath{\mathsf{st}}}
\newcommand{\Amem}{\ensuremath{\mathsf{mem}}}
\newcommand{\ARmem}{\ensuremath{\mathsf{a}}}

\newcommand{\Trans}{\ensuremath{\mathsf{trans}}}

\newcommand{\Parity}{\ensuremath{\mathsf{Parity}}}

\newcommand{\Del}{\ensuremath{\mathsf{Delete}}}
\newcommand{\Copy}{\ensuremath{\mathsf{Copy}}}

\newtheorem{problem}{Problem}

\newcounter{cntLemmaNumber}

\title{Deciding Graph Property by Single Mobile Agent: \\ One-Bit Memory Suffices\footnote{A part of this paper was presented in the 2018 ACM Symposium on Principles of Distributed Computing (PODC2018) as a brief announcement.}}

\author{
Taisuke Izumi\footnote{Graduate School of Information Science and Technology, Osaka University, E-mail: \{t-izumi, n-kitamura, masuzwa\}@ist.osaka-u.ac.jp.}
\and Kazuki Kakizawa\footnote{Graduate School of Engineering, Nagoya Institute of Technology, E-mail: \{kezawa0612, btk15049\}@gmail.com}
\and Yuya Kawabata\footnotemark[3]
\and Naoki Kitamura\footnotemark[2]
\and Toshimitsu Masuzawa\footnotemark[2] 
}

\date{}

\begin{document}

\maketitle
\begin{abstract}
We investigate the computational power of the deterministic single-agent model where the agent and each node are equipped with a limited amount of persistent memory. Tasks are formalized as decision problems on properties of input graphs, i.e., the task is defined as a subset $\mathcal{T}$ of all possible input graphs, and the agent must decide if the network belongs to $\mathcal{T}$ 
or not. We focus on the class of the decision problems which are solvable in a polynomial number of movements, and polynomial-time local computation. The contribution of this 
paper is the computational power of the very weak system with one-bit agent memory and $O(1)$-bit storage (i.e. node memory) 
is equivalent to the one with $O(n)$-bit agent memory and $O(1)$-bit storage. 
We also show that the one-bit agent memory is crucial to lead this equivalence: There exists a decision task which can be 
solved by the one-bit memory agent but cannot be solved by the zero-bit memory (i.e., oblivious) agent.
Our result is deduced by the algorithm of simulating the $O(n)$-bit memory agent by the one-bit memory agent with polynomial-time overhead, which is developed by two novel technical tools. The first one is a dynamic $s$-$t$ path maintenance mechanism which uses only $O(1)$-bit storage per node. The second one is a new lexicographically-ordered DFS algorithm for the mobile agent system with $O(1)$-bit memory and $O(1)$-bit storage per node. These tools are of independent
interest.
\end{abstract}

\section{Introduction}
\subsection{Background and Our Result}
The \emph{mobile agent} model is one of the popular computational models in distributed computing,
where autonomous entities (i.e., \emph{agents}) moving in the input network process some given 
task. It is recognized as an abstraction of robot systems, web crawlers, swarm intelligence, 
and so on. The models of mobile agents are characterized by several features such as  
the number of agents, capability of agents and nodes, faults, and network topology. It is 
a central question in the theory of mobile agents to clarify the computational power of 
a given model.

In this paper, we investigate the computational power of the deterministic 
single-agent model where the agent and nodes are equipped with a limited amount of 
persistent memory. Throughout this 
paper, we refer to the memory managed by the agent as ``memory'', and that managed 
by each node as ``storage''.  The detailed 
characteristics follow the standard assumptions of mobile agent systems: The input graph $G$ 
is \emph{anonymous} (i.e., there is no unique ID of nodes) and neighbors of a node is 
identified by \emph{local port numbers} assigned to the edges incident 
to the node. The agent does not have any prior knowledge on the network topology. In one time step, 
the agent performs local computation, which includes the update of the information stored in 
its own memory and the storage at the current node, and decides the neighbor to which it 
moves. 

We formalize tasks as decision problems on properties of 
input networks. More precisely, a task is a subset $\mathcal{T}$ of all possible 
inputs $\mathcal{G}$. The agent must decide if the input $G$ belongs to $\mathcal{T}$ or not, i.e., 
the agent terminates with one-bit output indicating $G \in \mathcal{T}$ or not. This decision 
problem is defined in a very general form and covers many popular graph problems.
For example, all the following problems belong to this class: planarity, bipartiteness, and 
boundedness of graph parameters such as diameter and treewidth.  

We focus on the class of the decision problems which are solvable by a single mobile 
agent in a polynomial number of movements and polynomial-time local computation. 
Such a class is a natural analogue of class $P$ in the centralized computation. Formally, we denote by $\Pagt(g(n), h(n))$ the class of the decision problems which are solvable in the mobile 
agent system with $g(n)$-bit memory and $h(n)$-bit storage. The main contribution of 
this paper is to reveal the relationship between function $g(n)$ and the 
task solvability. Surprisingly, we show that the magnitude of $g(n)$ does not matter. 
Even if $g(n) = 1$, the agent can use the asymptotically whole storage as its own memory.
We have the following theorem:
\begin{theorem} \label{thm:main}
$\Pagt(1, \Theta(1)) = \Pagt(\Theta(n), \Theta(1))$\footnote{The asymptotic version of this class $(\Pagt(\Theta(g(n)), \Theta(h(n))$ is defined as the union of $\Pagt(g'(n), h'(n))$ for all $g'(n) = \Theta(g(n))$ 
and $h'(n) = \Theta(h(n))$. }.
\end{theorem}
It is easy to show that our result is tight with respect to memory size. That is, the computational powers of the one-bit memory agent and the zero-bit memory agent (so-called \emph{oblivious} agent) are not equivalent under the assumption of $O(1)$-bit
storage. The following theorem is obtained as a corollary of the impossibility result by Cohen et al.~\cite{CFIKP08}.
\begin{theorem} \label{thm:impossible}
There exists a decision problem $\mathcal{P}$ such that 
$\mathcal{P} \not\in \Pagt(0, \Theta(1))$ and $\mathcal{P} \in \Pagt(1, \Theta(1))$ hold. That is, $\Pagt(0, \Theta(1)) 
\subsetneq \Pagt(1, \Theta(1))$.
\end{theorem}

\subsection{Technical Idea}

\sloppy{
The proof of Theorem~\ref{thm:main} consists of two parts, each of which shows
$\Pagt(O(1), O(1)) = \Pagt(O(n), O(1))$ and $\Pagt(1, O(1)) = \Pagt(O(1), O(1))$. 
The second part of reducing the agent memory size from $O(1)$ to $1$ is 
realized by a simple idea, which disassembles the information transfer caused by 
the movement of $k$-bit memory agent into $k$ iterations of round-trip movement of 
the one-bit agent. Hence the main technical challenge is at the first part. The key ingredient of the first part is an algorithm for serializing all the nodes in the network. The serialized nodes, each of which has $O(1)$-bit storage, 
can be treated as the $O(n)$-bit tape $T$ of a Turing machine(TM), and our simulation algorithm utilizes it as the agent 
memory. Letting $A$ be any algorithm utilizing $O(n)$-bit agent memory, 
we model $A$ (i.e., the algorithm of the simulated agent) as a Turing machine $M$ utilizing $O(n)$-bit tape, and the simulation algorithm executes $M$ on tape $T$. The primary 
hurdle of this approach is that we cannot explicitly store the serialized order of 
the nodes. More generally, $O(1)$-bit storage per node does not suffice to store any 
\emph{path structure} assigning each node with a predecessor neighbor and a successor 
neighbor explicitly. The main technical tool of our solution is to provide a new distributed data structure \emph{R-path}, which dynamically maintains a traversable path 
between two given nodes by spending only $O(1)$-bit storage per node and 
$O(1)$-bit memory. Utilizing this technique, our simulation algorithm implements the lexicographically-ordered depth-first search (Lex-DFS) for the input network $G$. It 
obviously provides an implicit serialization of all nodes following the pre-order of Lex-DFS. 
When simulating $M$ using $O(n)$-bit tape, the simulator agent always keeps the 
state of $M$ in its memory. In addition, a constant number of R-path instances are used to connect the locations of the simulator agent, simulated agent, and tape head. 
Then the agent can freely access the information stored in those locations by traversing the maintained paths, and can simulate one-step execution of $M$. 
}

Finally, we emphasize that our Lex-DFS algorithm itself is also a by-product contribution of this paper. To the best of our knowledge, it is the first DFS algorithm only spending $O(1)$-bit memory and $O(1)$-bit storage per node. Since Lex-DFS has a wide range of applications, we believe that this tool will be useful for
importing such applications into the mobile agent systems with limited memory/storage capability.  

\subsection{Related Work}
To the best of our knowledge, there exists no prior work on addressing the 
space complexity of general tasks in the single-agent model. Nevertheless, the space 
complexity of the single-agent model for a specific problem is much investigated.
The mainstream of this research line is the \emph{graph exploration problem}, which 
has a long history and the origin goes back to the Shannon's experiment on maze-solving mouse~\citep{clau}. In the field of theoretical computer science, it received much attention 
after the seminal paper by Aleliunas et al.~\citep{romas} providing 
the explicit polynomial-time upper bound of the random walks in undirected graphs
(so-called \emph{cover time}). It implies that a single $O(\log n)$-bit counter yields 
a simple randomized graph exploration algorithm with termination. This result triggers the interest of log-space \emph{deterministic} graph exploration algorithms in relation to the 
computational complexity issue of log-space computability, and a number of algorithms are developed~\citep{mich,shlo,michal2003,sorin1988,howard1988}. The Reingold's seminal logspace undirected connectivity algorithm is one of the milestones on this line~\citep{gold2008}. 
The memory-size optimality of the Reingold's algorithm is also shown by Fraigniaud et al.~\citep{fraign2005}. The speed matter of graph exploration is considered with 
respect to time-space tradeoff~\citep{uriel1997,adrian2013}.

While all the results above assumes the system without storage, sublogarithmic-memory 
graph exploration algorithms assuming storage has been proposed, for a variety of storage models. 
For example, the \emph{rotor-router} model~\citep{bamp}, which is also known as a technique of de-randomizing random walks, guides the agent by the pointer managed by each node. This model attains the graph exploration with zero-bit memory and $O(\log \Delta)$-bit storage, where $\Delta$ is the maximum degree of the network. Yet another example is the 
the pebble-based algorithm proposed by Disser et al.~\citep{disser}, which is a graph exploration algorithm using $O(\log\log n)$ distinct pebbles. The $k$-pebble model 
is one of the restricted versions of the standard storage model, where each 
agent has $k$ distinguished pebbles, and can put and pick-up each pebble on 
any node for leaving some information there. The conceptual idea of this algorithm
is very close to our result. The algorithm simulates a log-space Turing machine
in a distributed manner using $O(\log\log n)$ pebbles, and runs the Reingold's algorithm
on the simulated Turing machine. Since this model is incomparable to our $O(1)$-bit storage model, the technical contribution
of our result is definitely independent of this algorithm. In 
addition, the simulation overhead of the algorithm by Disser et al. is unfortunately super-polynomial, while our result only
incurs the polynomial-time overhead. The polynomial-time graph exploration using $O(1)$-bit memory and $O(1)$-bit storage 
has been shown to be solvable by Cohen et al.\cite{CFIKP08}. The exploration algorithm
presented in that paper is actually the breadth-first search (BFS) by the agent.
Our DFS algorithm complements another side of very standard graph search algorithms. 
Note that it is certainly possible to use the exploration algorithm by Cohen et al. as the building block of our 
simulation algorithm for node serialization, instead of our lex-DFS algorithm. However, no general simulation technique 
as our result is addressed there. The paper by Cohen et al.\cite{CFIKP08} also presents the impossibility
of graph exploration by zero-bit memory agents if the storage size is $o(\log \Delta)$ bits, which is the basis of 
the proof of Theorem~\ref{thm:impossible}.
A variety of 
the exploration algorithms attaining different memory/storage complexities are also 
presented in \cite{SBNOKM15}, but no algorithm there attains the complexity of 
$O(1)$-bit memory and $O(1)$-bit storage.

\subsection{Roadmap}
The paper is organized as follows: In Section~\ref{sec:prelim}, we introduce the formal 
model and the problem definition.
Section~\ref{sec:main} explains the Lex-DFS algorithm using $O(1)$-bit memory and $O(1)$-bit storage, as well as
the details of the R-path structure.
Following the proof of $\Pagt(\Theta(1),\Theta(1)) = \Pagt(\Theta(1),\Theta(n))$ in Section~\ref{sec:sim-constant-bit}, 
we present the result of $\Pagt(1,\Theta(1)) = \Pagt(\Theta(1),\Theta(1))$ in Section~\ref{sec:sim-one-bit}. The proof of Theorem~\ref{thm:impossible} 
is straightforward, and thus it is deferred to the appendix.

\section{Preliminaries}\label{sec:prelim}

\subsection{Input Graph}
We denote by $[a,b]$ the set of integers at least $a$ and at most $b$. The input graph 
$G = (V, E, \Pi)$ is a simple undirected connected graph with port numbering, where 
$V$ is the set of nodes, $E \subseteq V \times V$ is the set of edges, $\Pi$ is the set of 
port-numbering functions (explained later). We denote 
by $n$ and $m$ the numbers of nodes and edges respectively. Let $E(v)$ be 
the set of edges incident to node $v$, $N(v)$ be the set of nodes 
that are adjacent to $v$, and $\Delta(v)$ be the degree of node $v \in V$. While
each element $v$ in $V$ is treated as the unique ID of the corresponding node, it is not 
identified by the agent. That is, nodes are \emph{anonymous}. For the agent to 
distinguish adjacent nodes, a \emph{local port numbering} is defined for each node. An edge 
incident to a vertex $v$ is assigned with a unique label in $[0, \Delta(v) - 1]$ defined by a port numbering function $\pi_v\colon E(v) \to [0, \Delta(v) - 1]$. The set $\Pi$ consists of 
the port-numbering functions for all nodes in $V$. Note that two port numbers $\pi_v(e)$ and 
$\pi_u(e)$ assigned to the same edge $e = (u, v)$ are independent. Since $E(v)$ and $N(v)$ has one-to-one 
correspondence, we often abuse $\pi_v$ as a function from $N(v)$. That is, for any $u \in N(v)$, $\pi_v(u)$ represents the $v$'s port number of edge $(v,u)$. Furthermore, we also denote the inverse function of $\pi_v(u)$ on $N(v)$ by $\pi^{-1}_v$ (i.e., $\pi^{-1}_v \colon [0, \Delta(v) - 1] \to N(v)$). 

\subsection{Mobile Agent}
Each node $v \in V$ has \emph{storage}, which is a persistent memory keeping 
information even after the agent leaves (i.e., the agent can refer to that information when 
it comes back). We denote by $b(v)$ the information stored in the storage at $v \in V$. Since we consider the system with
 $O(1)$-bit storage per node, $b(v)$ is treated as a fixed-length bit string. At the beginning of 
the execution, $b(v)$ for all $v \in V$ are initialized with the sequence of zeros.
The execution of the agent follows discrete time steps. In each step, 
the agent first observes the degree of the current vertex $v$, the port number of the edge from which it entered $v$, and the storage of $v$. Then it 
performs local computation following the algorithm and the observed information. The local computation outputs the information left to the storage of the current node and the port number $x \in [0, \Delta(v) - 1]$ indicating 
the destination of the next movement. Finally, the step finishes with the movement to $\pi^{-1}_v(x)$. It is 
guaranteed that the movement at a step finishes by the beginning of 
the next step. The location of the agent at the current step is 
referred to as {\Vcur}, and the port number from which the agent comes to {\Vcur} is referred to as {\Portin}. Note that the ``current step'' implied by notations {\Vcur} and {\Portin} depends on the context.

An algorithm $A$ is formally defined by a 5-tape Turing machine $M(A) = (Q, \Sigma, \gamma, q_{0}, q_{1})$, where $Q$ is the set of states, $\Sigma = \{0, 1, \sqcup \}$ is the tape alphabet, $\gamma \colon Q \times \{0,1\}^5 \to Q \times \{0, 1\}^5  \times \{L,R\}^5$ is the transition function, 
$q_{0}$ is the initial state, and $q_{1}$ is the termination state. In what follows, 
we refer to an atomic step of TMs as ``TM step'', and refer to the time step of agent 
executions as ``agent step'' or simply ''step'' for distinguishing them clearly. That is, 
one agent step consists of multiple TM steps of $M(A)$ and the following agent movement after
the termination of $M(A)$. The first tape corresponds to the local persistent memory of 
the agent, which is not reset even after the TM terminates and thus accessible in the next agent step. On the other hand, at the beginning of every agent step, the 2-4th tapes are initialized with the information of $b(\Vcur)$, the degree $\Delta(\Vcur)$, and $\Portin$ respectively. The 2-4th tapes are read-only in executions of $M(A)$. The fifth tape is write-only, which initially stores only symbol $\sqcup$ and finally stores the port number from which the agent leaves $\Vcur$ at the termination of $M(A)$. Due to some technical reason, we assume that the information stored in the 3-5th tapes are unary encoded, i.e., a numerical value $k$ is encoded by $k$ iterations of 1. Note that this assumption is not essential because unary and binary encodings are mutually transformable in polynomial TM steps. The agent performs the local computation following the transition function $\gamma$ until $M(A)$ terminates at state $q_1$. The algorithm $A$ terminates when the local computation finishes
with writing nothing in the fifth tape (i.e., storing only $\sqcup$). The maximum lengths 
$g(n)$ and $h(n)$ of the first and second tapes 
used by $M(A)$ over all possible executions for graphs of at most $n$ nodes are respectively called 
the \emph{memory complexity} and \emph{storage complexity} of $A$.  
We call $A$ a \emph{polynomial-time} algorithm if for any input graph $G$ of $n$ nodes, 
$A$ necessarily terminates within $\Poly(n)$ agent steps and local computation by $M(A)$ 
necessarily terminates within $\Poly(n)$ TM steps.

\subsection{Decision Task and Complexity Class}
Let $\mathcal{G}$ be the set of all possible input graphs. A \emph{decision task} $\mathcal{T}$ is defined as 
a subset of $\mathcal{G}$ consisting of all ``yes'' instances of that task. Assume that the storage at each node
contains a boolean output register. An algorithm $A$ solves a decision task $\mathcal{T}$
if for any input graph $G \in \mathcal{G}$ and any initial location of the agent, the agent 
terminates with writing the value of the predicate $G \in \mathcal{T}$ 
in the output register at the final location. We define $\Pagt(g(n), h(n))$ be the class of the decision tasks 
solvable by a polynomial-time algorithm of memory complexity $g(n)$ and storage complexity 
$h(n)$.

\section{Lex-DFS Algorithm Using $O(1)$-bit memory and $O(1)$-bit storage} \label{sec:main}

\subsection{Lexicographically-Ordered DFS}
The lexicographically-ordered depth-first search (Lex-DFS) is a special case of the standard depth-first search, which requires some specific search order following the port numbers: 
When the agent chooses an unvisited neighbor in $N(\Vcur)$, 
it must choose the node $u \in N(\Vcur)$ whose port number $\pi_{\Vcur}(u)$ is the minimum of all unvisited neighbors. Throughout this paper, we fix the root of the 
Lex-DFS search by the initial location of the agent, which is denoted by $r$. We also 
define $r = v_0, v_1, \dots, v_{n-1}$ as the sequence of all vertices in the input graph $G$ following the pre-order of 
Lex-DFS. Fixing the root $r$, the Lex-DFS ordering and the corresponding DFS tree is uniquely determined. We denote 
by $P(v)$ the path from $r$ to $v$ in the Lex-DFS tree.

\subsection{Procedure \textsf{FindFirst}} 
For ease of presentation, we introduce a fundamental subroutine called \textsf{FindFirst}. 
In the design of our algorithm, the agent often has to discover the neighbor with a specific state. 
The procedure \Findfirst$(O,P)$ invoked at a node $u$ visits all the neighbors $N(u)$ in the order specified by $O$
for probing. When the agent finds the neighbor $v \in N(u)$ 
satisfying the predicate $P$ on the storage contents for the first time, it goes back to $u$ and terminates.
Then $\Portin$ indicates the port number of $v$, i.e., it terminates with $\Portin = \pi_u(v)$. 
To simplify the description of the algorithm , if there are no neighbor $v$ satisfying 
$P(v)=\True$, \Findfirst$(O, P)$ terminates with $\Portin = -1$~\footnote{Since the 
system does not have port number $-1$, it is actually implemented by preparing one-bit flag
indicating the success or failure of the exploration in the agent state.}. 
We can choose the search order $O$ from the following four options: HeadAscend,
MiddleAscend, MiddleDescend, TailDescend.
The choice of Head, Middle, and Tail determines the
port number where the procedure starts the search, which respectively means
$0$, $\Portin$, $\Delta(u) - 1$. The choice of Ascend/Descend is 
the order of the search, each of which corresponds to the ascending and descending 
orders of port numbers. The search with option Middle is not cyclic, i.e.,
it terminates when the neighbor with port number $0$ or $\Delta(u) -1$ has been probed.
Note that iterative probing of neighbors does not need any persistent counter: In the
invocation at node $v$, the agent repeats the go-and-back for each neighbor in $N(v)$,
which is implemented by the mechanism that the agent coming back to $v$ moves to the
neighbor through port $\Portin + 1$ or $\Portin - 1$ according to the choice of Ascend/Descend.

\subsection{Lex-DFS Algorithm Using $O(1)$-Bit Memory and $O(1)$-Bit Storage} \label{subsec:lexdfs}
This section provides the details of our small-space Lex-DFS algorithm, called \textsf{DLDFS}. 
In the following argument, we use notation $b(v).\mathsf{q}$ to represent the variable $\mathsf{q}$ 
stored in the storage $b(v)$. The fundamental structure of 
\textsf{DLDFS} follows the centralized Lex-DFS algorithm using only $O(n)$-bit memory space 
(more precisely, it manages vertex labels of $O(1)$ bits) proposed concurrently by 
Asano et al.~\citep{asano} and Elmasry et al.~\cite{elma}. Initially, all vertices are 
colored by white (implying ''not visited yet'').  All visited nodes are colored by black (implying ``already visited'') or grey (implying ``already visited, and on the path 
$P(u)$ for the location $u$ of the DFS search head''). We refer to $P(u)$ as the \emph{grey path}. In \textsf{DLDFS}, the color
of each node is managed by the storage variable $\mathsf{color}$.
In the forwarding phase, \textsf{DLDFS} behaves completely same as the naive Lex-DFS algorithm. 
First, it checks if there exists a white neighbor of $\Vcur$ or not.
If it exists, the algorithm chooses the one with the minimum port number as the next node,
and colors the next node by grey. This process is easily implemented by the procedure $\textsf{FindFirst}$. 
When no white neighbor exists, the algorithm performs the backtrack to the parent node in the grey path, 
as well as coloring the current node by black. Since storing the pointers 
to the parent node at each node on $P(u)$ incurs $O(\log n)$-bit storage cost, the technical challenge of 
$O(1)$-bit storage Lex-DFS lies at how to perform backtracks without those pointers. 
The solution by Asano et al. and Elmasry et al. is to resolve this mater by the following lemma, which states 
that the grey path from $r$ is recovered only from the information of the \emph{set} of 
grey nodes.

\begin{lemma}[Asano et al.~\citep{asano}, Elmasry et al.~\citep{elma}] \label{lma:stackcorrectness}
Let $G = (V, E, \Pi)$ and $r \in V$ be any instance of Lex-DFS, and $ P(u) = u_0, u_1, 
\cdots, u_k$ ($r = u_0$ and $u = u_k$). 
Then, for any $i \in [0, k-1]$, $\pi_{u_i}(u_{i+1}) = 
\min \{ \pi_{u_i}(u_j) \mid j \in [i + 1, k], u_j \in N(u_i) \}$ holds.
\end{lemma}

The lemma above deduces the following backtrack process at node $u$: First, we color $u$ by an extra color (we use green
for this purpose). 
Then the agent starts the traversal of $P(u)$ from the root node $r$. 
Consider the situation where $u_0, u_1, \cdots, u_{i}$ are already 
traversed and the algorithm at the current node $u_i$ is finding the next node $u_{i+1}$. Assume that each node maintains 
the boolean flag $\mathsf{traversal}$ in its storage, and all traversed nodes set {\True} to that flag. Then the algorithm can identify $u_{i+1}$ because it is the non-flagged grey neighbor of $u_i$ having the minimum port number, by Lemma~\ref{lma:stackcorrectness}. 
This search task is implemented by {\Findfirst}. When $u_{i+1}$ is found, $\Portin$ indicates the port to $u_{i+1}$. 
The agent visits $u_{i+1}$ and sets $b(u_{i+1}).\mathsf{traversal} = \True$. Iterating this process the algorithm 
completely traverses the sequence $P(u)$ and can identify the parent of $u$ (when {\Findfirst} finds a green neighbor,
the current node is the parent). After traversing $P(u)$, all flags must be reset,
which is done by executing the traversal from $r$ again by inverting the roles of $\mathsf{traversal}$
(i.e., $\mathsf{traversal} = \True$ represents that it is not flagged). Finally, the backtrack finishes with coloring $u$ by black. 

The remaining issue is how the agent goes back 
to the root $r$ to start the traversal of $P(u)$. To address it, we introduce 
a new distributed date structure called \emph{R-path}. An instance $X$ of R-path 
maintains a path connecting an arbitrary  \emph{source node} $s$ to 
a \emph{target node} $t$. Note that R-path structure does not store any specific 
path from $s$ to $t$, but some path connecting from $s$ to $t$ admitting the 
traversal by the agent. In that sense, what the R-path structure provides is only the functionality of jumping between $s$ and $t$. The important features of 
R-path are twofold: We can dynamically update the location of the target node $t$, 
and $O(1)$-bit memory and storage suffice to implement it. Each instance $X$ of the R-path structure provides the following three operations to the upper application layer. Note that every operation can be executed only when the agent is on the managed 
path, and {\Modifymove} can be executed only when the agent is on the target node $t$.
\begin{itemize}
\item \Movetop: The agent on the path managed by $X$ goes back to the source node of $X$. 
\item \Modifymove: The agent changes the target node of $X$ to $\pi^{-1}_{\Vcur}(\Portin)$, and moves there.
\item \Moveonehopdown: When the agent is on the path managed by $X$, it moves 
to the neighbor in that path closer to the target node .
\end{itemize}
Apparently, the R-path structure is the last piece of completing \textsf{DLDFS}. In running the Lex-DFS algorithm explained above, the agent manages the path from the root $r$ to the tail of the current grey path using the R-path structure. When it must go back to $r$ in the process of backtracks, it suffices to run {\Movetop}. Assuming the R-path structure is realizable, 
we obtain the following theorem.

\begin{theorem} \label{thm:mainDFS}
There exists a polynomial-time Lex-DFS algorithm using $O(1)$-bit memory and $O(1)$-bit storage.
\end{theorem}

In the following applications, we often run \textsf{DLDFS} multiple times. Then 
all nodes are colored by black after the first invocation of the algorithm. To run \textsf{DLDFS} again, those colors must be reset to white. This matter is easily 
resolved in the same way as resetting all flags in the grey path traversal. That is, 
it suffices to run \textsf{DLDFS} again by swapping the meanings of black and white nodes.


\subsection{Implementation of R-path Structure}
We first explain the structure of memory and storage:
The agent is oblivious in procedures {\Movetop} and {\Moveonehopdown}. 
The implementation of {\Modifymove} uses two agent states, called \emph{Find} and 
\emph{Delete}. The storage consists 
of four variables \target, \inPath, \direction, and \Color. Note that the variable {\Color} in this section is independent of the variable with the same name in the main Lex-DFS algorithm.
The variable {\target} is a flag indicating the current target node. 
The variables {\inPath} and {\direction} are also the binary flags for recording the 
information on the maintained path (the details are explained later). The 
variable {\Color} is a mark internally used in the procedures,
which takes one of four colors $\{\mrm{white}, \mrm{red}, \mrm{blue}, \mrm{yellow}\}$.
Initially, all nodes have color $\mrm{white}$.
We call the node satisfying {\inPath} $=\True$ an \emph{in-path node}, and denote by $P$ 
the set of all in-path nodes. The main idea for saving space is that R-path maintains the
set of in-path nodes constituting a ``minimal'' path from the source node to the target node, 
where the minimality means that any subgraph induced by $P$ forms a path graph.
This feature guarantees that any node in $P$ has at most two neighbors in $P$. 
We call the neighbors of $v$ in $P$ closer to $s$ and $t$ the \emph{backward} and \emph{forward} neighbors of $v$
respectively, which are denoted by $\Pred(v)$ and $\Succ(v)$. To traverse the maintained path correctly, 
we further add one-bit information by variable {\direction}. Precisely, 
the value {\UlargerD} (resp. {\UsmallerD}) of variable {\direction} implies 
that the port number of the backward neighbor is greater 
(resp. smaller) than that of the forward one. Supported by this information, 
the agent can perform the backward and forward movement in the path correctly utilizing the procedure {\Findfirst}.  For example, if $\b(\Vcur).\direction = \UsmallerD$ and the agent wants to find $\Pred(\Vcur)$, the agent runs {\Findfirst} with predicate $b(\cdot).\inPath = \True$ and \textsf{HeadAscend} order. Then {\Findfirst} returns the port number of $\Pred(\Vcur)$. 
All other cases can be processed similarly (see Figure~\ref{fig:predsucc}).
It is straightforward to implement the procedures {\Movetop} and {\Moveonehopdown} utilizing this mechanism. 

\begin{figure}[t]
\begin{center}
\includegraphics[width=150mm]{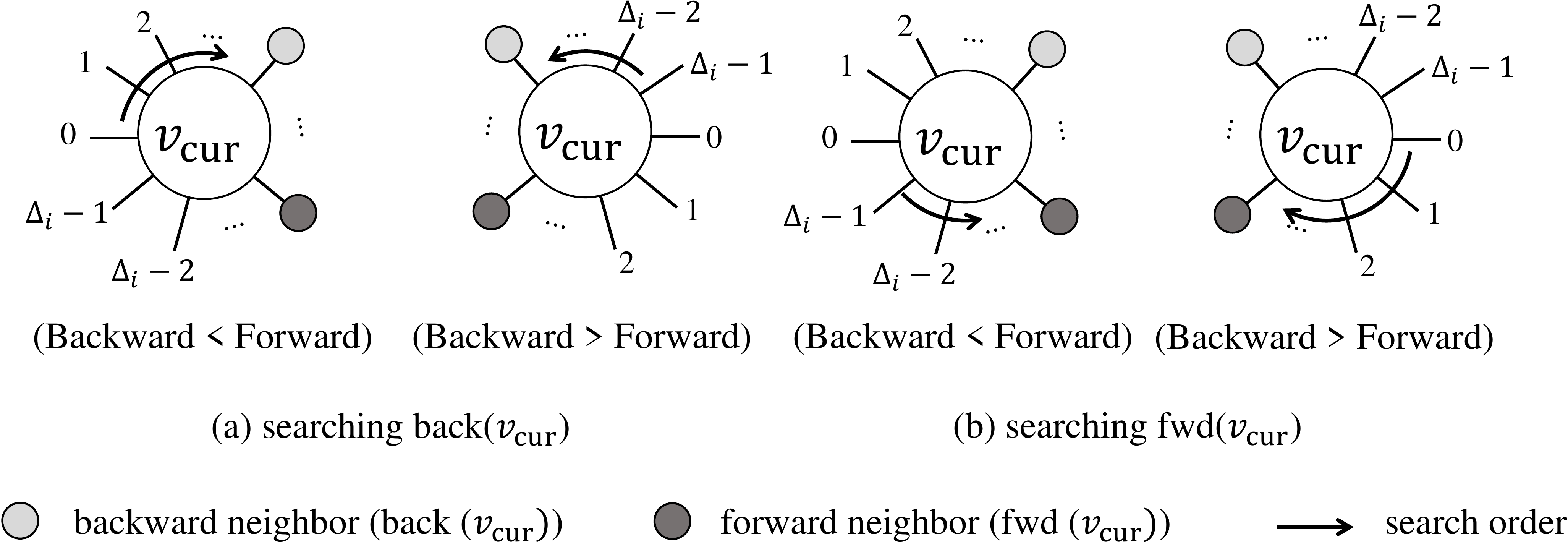}
\caption{Finding $\Pred(v)$ or $\Succ(v)$}
\label{fig:predsucc}
\end{center}
\end{figure}

The main technical challenge of realizing R-path is how to implement
\Modifymove. Consider the situation where the agent is on the target node $t$ and wants to update 
the target node with a neighbor $t' \in N(t)$. First, the agent colors 
$t'$ with \mrm{yellow}. If $t'$ is on the maintained $s$-$t$ path, the edge $(t',t)$ is the tail edge of that path (Figure~\ref{fig:modifymove}(a)). Then the update completes by simply 
removing it, i.e., setting $b(t).\inPath = \False$, $b(t).\target = \False$, $b(t').\Color = \mathrm{white}$, and 
$b(t').\target = \True$. The case that $t'$ is not on the maintained $s$-$t$ path is more complicated because 
the resultant $s$-$t'$ path might not be minimal. Then, the agent moves 
to the source node $s$ by invoking {\Movetop}, and starts the \emph{find phase} by changing its state to Find. 
In the find phase, the agent descends the current $s$-$t$ path by using 
{\Moveonehopdown} (Figure~\ref{fig:modifymove}(b1)). During the phase 
it also checks if the current in-path node has $t'$ as its neighbor, which is
done by searching the yellow neighbor with {\Findfirst}. Since the agent keeps moving forward until $t'$ is found
in the current neighbors and $t'$ is found at latest, the find phase always terminates with successfully finding 
the yellow neighbor. 
If the yellow neighbor is found at a node $u$ (referred to as the \emph{branching node}) for the first time, 
the concatenation of the path from $s$ to $u$ and the edge $(u, t')$ creates 
a new path connecting $s$ and $t'$, which satisfies the minimality requirement 
of R-path. After finding the branching node $u$, the agent must continue to descend 
the $s$-$t$ path to $t$ for deleting the expired path from $u$ to $t$. To obtain the forward way to $t$ at node $u$, 
we do not immediately modify the variable {\direction} of $u$. Instead, when $u$ is found, 
the agent colors node $u$ by \mrm{red} or \mrm{blue}. The color \mrm{red} (resp. \mrm{blue}) implies that the variable {\direction} at node $u$ must be updated by the value 
{\UlargerD} (resp. {\UsmallerD}). 
The color assigned to $u$ can be computed by comparing the port
numbers of the yellow neighbor and $\Pred(u)$ using {\Findfirst}, but it is not an obvious task because one cannot distinguish $\Pred(u)$ and $\Succ(u)$ during the run of {\Findfirst}. We resolve this matter by the following strategy: The agent first 
finds the port number of the yellow neighbor by {\Findfirst} with option HeadAscend, and then starts {\Findfirst} again
with Middle- option from the position of the yellow neighbor to find any in-path neighbor. The search order of the second invocation 
follows the ascending order if $b(u).\direction = \UlargerD$ holds, or the descending order otherwise.
The agent decides the color of $u$ by blue if one of the following cases applies: 
(1) In the ascending order search, no in-path neighbor is found, or (2) in the descending 
order search an in-path neighbor is found. The node $u$ is colored by red in all other cases. 
We can see that the procedure correctly decides the color of $u$ from Figure~\ref{fig:branching} that 
describes all the possible cases.
The actual updates of $b(u).\direction$ is deferred to the end of the following
\emph{delete phase}. In the delete phase, the agent continues to 
descend the path with resetting {\inPath} by \False (Figure~\ref{fig:modifymove}(b2)). Arriving 
at $t$, the agent finds the yellow neighbor $t'$ and moves to there. The agent sets
$b(t').\Color = \mathrm{white}$ and $b(t').\inPath = \True$, and then further moves to the neighbor with colored 
by red or blue (i.e., $u$). The variable $b(u).\direction$ 
is appropriately updated according to its color. Finally, the agent goes back to $t'$ and the procedure finishes (Figure~\ref{fig:modifymove}(b3)). The pseudocode of the procedure 
{\Modifymove} and the correctness of the R-path structure is given in the appendix. 

\begin{figure}[t]
\begin{center}
\includegraphics[width=150mm]{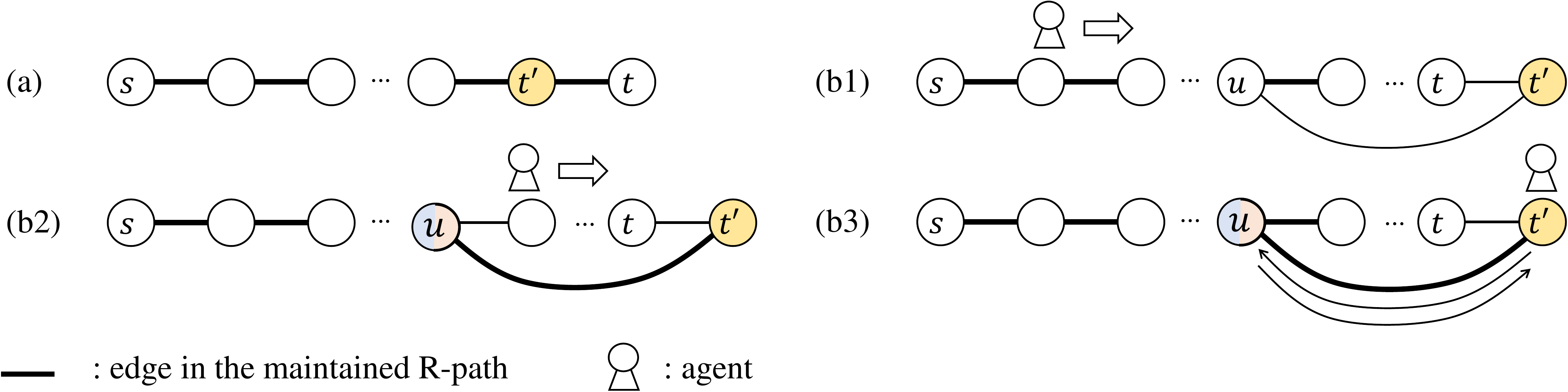}
\caption{The implementation of {\Modifymove}}
\label{fig:modifymove}
\end{center}
\end{figure}

\begin{figure}[t]
\begin{center}
\includegraphics[width=150mm]{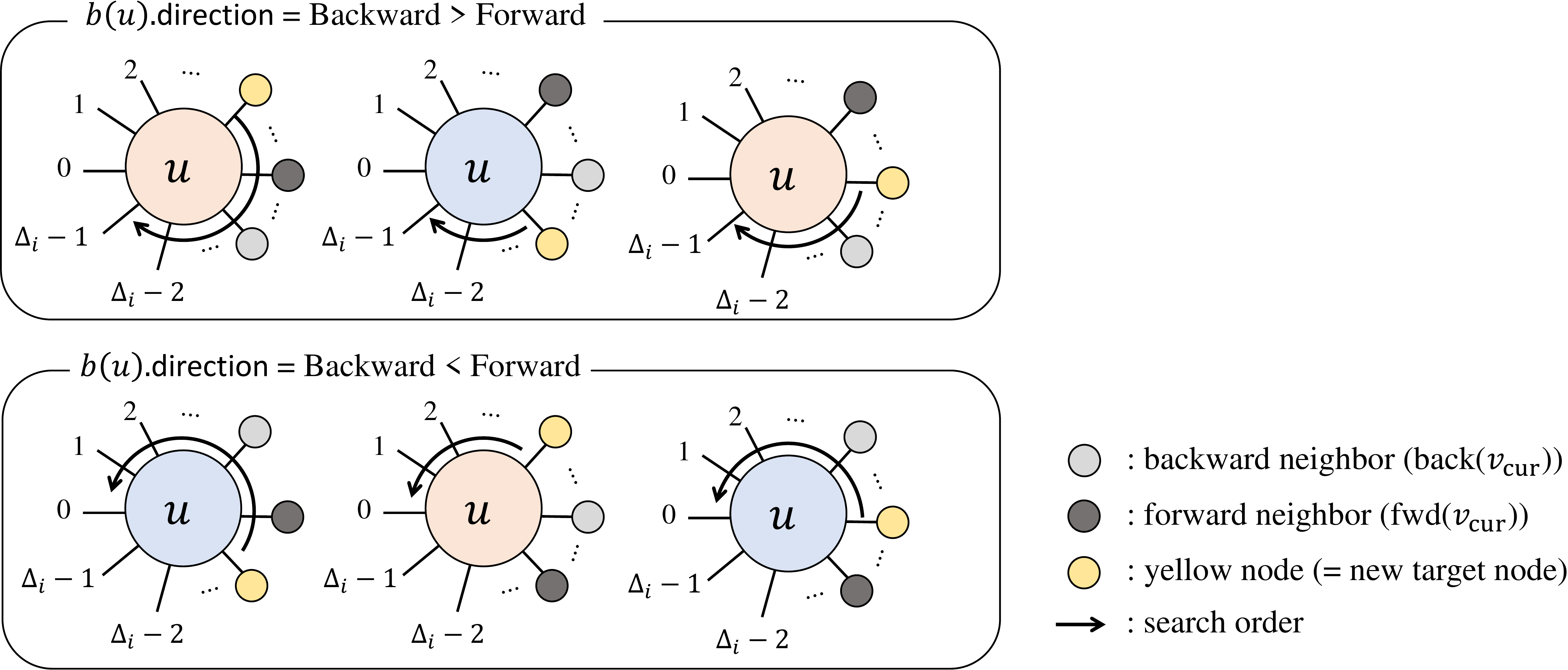}
\caption{Coloring of branching nodes. Recall that the goal of this task is to color $u$ by red if the port number of the yellow neighbor precedes that of the backward in-path neighbor, or by blue otherwise.}
\label{fig:branching}
\end{center}
\end{figure}

\section{Simulation of $O(n)$-bit Agent by One-bit Agent} \label{sec:simulation}

\subsection{Extending R-path structure}
\label{sec:ex_Rpath}
We first introduce four additional operations for the R-path structure. 
\begin{itemize}
\item \Movetarget: The agent goes down to the target node maintained by the instance.
\item \Del: Change the current target node to the current source node.
\item \Copy($X_1, X_2$): Given two instances $X_1$ and $X_2$ of the R-path, the path maintained by $X_1$ is copied 
to the instance $X_2$. That is, by this operation, $X_1$ and $X_2$ maintains the same path. The operation terminates
at the target node of $X_1$ (which is equal to that of $X_2$).
\item \Modifysucc / \Modifypred: Change the current target node of the instance to its
successor/predecessor in the lex-DFS ordering. 
\end{itemize}
The procedure {\Movetarget} is easily implemented by iterative invocations of \Moveonehopdown.
The procedure {\Del} is also easily implemented by the following operations: The agent executes {\Movetarget} first, and then 
iteratively invokes {\Modifymove} to change the target to its backward neighbor until the target becomes $r$.
The backward neighbor is easily found by by running {\Findfirst}. 

The procedure $\Copy(X_1,X_2)$ first invokes 
{\Del} for $X_2$, and then iteratively invokes
{\Modifymove} for $X_2$ to move the target of $X_2$ along the path maintained by $X_1$. That is, assuming the target node $t$ of $X_2$ is on the path maintained by $X_1$, the procedure finds the forward neighbor $v$ of $t$ in the path managed by 
$X_1$ using {\Findfirst}, and changes the target node of $X_2$ to $v$ by {\Modifymove}. 
When the target node of $X_2$ reaches the target node of $X_1$, the copy from $X_1$ to $X_2$ completes.

The implementation of {\Modifysucc} and {\Modifypred} requires 
a slightly non-trivial idea. It is realized by some modified version of \textsf{DLDFS} which manages 
an extra R-path instance $X'$. Let $I_j$ be the time-step interval in the execution of 
\textsf{DLDFS} from the first visit of $v_j$ to that of $v_{j+1}$. During the period $I_j$, 
$X'$ manages the path connecting $r$ and $v_j$. Consider 
the invocation of {\Modifypred} for $X''$ with the current target node $v$. It runs {\Movetop} for $X''$ to go back to the root $r$, and then starts \textsf{DLDFS}. At the beginning of each $I_{j + 1}$, the agent checks if $v_{j+1} = v$ holds or not. 
If it holds, the target node of $X'$ is $v_j$. Hence it suffices to invoke 
$\Copy(X', X'')$ to change the target node of $X''$ to the predecessor of $v = v_{j+1}$. Then the agent updates $X'$, 
which is done by copying the R-path instance $X$ managed inside of \textsf{DLDFS}, which has the target node 
$\Vcur = v_{j+1}$, to $X'$. Finishing the run of \textsf{DLDFS}, the procedure terminates with the invocation of {\Del} of 
$X'$. The procedure {\Modifysucc} is implemented similarly. 

\subsection{From $O(n)$ bits to $O(1)$ bits}
\label{sec:sim-constant-bit}

For simplicity, we construct the simulation algorithm $\Sim(A)$ for any algorithm $A$ using 
$n$-bit memory and $O(1)$-bit storage, but it is easy to extend $\Sim(A)$ for supporting general $O(n)$-bit memory.
As stated in the introduction, the key idea of the simulation is the distributed management of the tapes of $M(A)$ 
using the serialization of all nodes by \textsf{DLDFS}. To refer to the status of the simulated algorithm $A$, we use the notations with upper script $A$, e.g., $\Vcur^A$ denotes the 
current location of the simulated agent running $A$. The storage $b^A(v)$ of node $v$ in $A$ 
is also stored at $v$ in $\mathit{Sim}(A)$. The $i$-th node $v_i$ in the lex-DFS ordering stores the $i$-th bit of each tape.
Hence we treat the vertices $v_0, v_1, \dots, v_{n-1}$ as the address of each tape, i.e., we say ``the $j$-th tape head is at $v_i$'' to imply the $j$-th tape head is at
the $i$-th entry of the tape. The location of the $j$-th tape head of $M$ is denoted by
$t_j$ ($1 \leq j \leq 5$). To keep the information of those locations, the algorithm $\mathit{Sim}(A)$ manages six 
independent instances of the R-path, each of which is denoted by $X(z)$ for $z \in \{\Vcur^A, t_1, t_2, t_3, t_4, t_5 \}$
and manages a path from $v_0$ to $z$. In addition, each node has one-bit flag 
$\mathit{past}$ in its storage, which becomes {\True} if and only if the node is the 
location of the simulated agent one-step before. The state of $M(A)$ is stored in the memory of the simulator agent running $\mathit{Sim}(A)$. Obviously, the structure explained above
incurs only $O(1)$-bit memory and $O(1)$-bit storage per node. 
The one-step execution of $A$ performed at $\Vcur^A$ is simulated by the following phases. 
\begin{itemize}
    \item Initialization phase : As defined in Section~\ref{sec:prelim}, the 2-4th tapes must be initialized by the information of $b^A(\Vcur^A)$, the degree of $\Vcur^A$, and $\Portin^A$. 
    Copying the information $b^A(\Vcur^A)$ into the second tape is straightforward:
    The simulator agent stores whole information of $b^A(\Vcur^A)$ in its memory and takes it to
    the second tape by invoking {\Movetop} for $X(\Vcur^A)$ and {\Movetarget} for $X(t_2)$. The initialization of the third and fourth tapes are slightly different because the 
    information to be copied is not explicitly stored in $\Vcur^A$. Recall that the information in these two tapes is unary encoded. To read the information, 
    we use a similar technique as {\Findfirst}, where the agent probes all neighbors of $\Vcur^A$ in the ascending 
    order of their port numbers with tracking $\Vcur$ by an extra R-path instance $X(\Vcur)$. Consider an example of copying the information of the degree of $\Vcur^A$ to the third tape. The agent appends 1 at the next position of $t_3$ when it visits 
    one neighbor of $\Vcur^A$, by traversing the paths managed by $X(\Vcur)$ and $X(t_2)$. After probing all neighbors, 
    the third tape obviously stores the unary encoding of the degree information. The initialization of $t_4$ is also handled similarly (the details are deferred to the appendix). 
    \item Local simulation phase : The simulator agent iteratively simulates one TM step of 
    $M$ until $M$ terminates. The agent first collects 
    the values of the entries pointed by the tape heads. It is easily processed by the procedures {\Movetop} and {\Movetarget}. The simulation of one TM step is the local computation of the agent. After the computation, the agent must update the tape contents, and move the tape heads. It is also easily implemented by {\Modifysucc} and {\Modifypred}.
    \item Agent movement phase: When $M$ terminates, the first tape stores the state of the simulated agent after the local computation, and the fifth tape stores the port number to which the simulated agent goes out from $\Vcur^A$. Updating $\Vcur^A$ is seen as the ``inverse'' operation of initializing the fourth tape. Hence one can implement it in a 
    similar way as the initialization phase (the details are deferred to the appendix). 
 \end{itemize}

\subsection{From $O(1)$ bits to One bit}
\label{sec:sim-one-bit}
Let $A$ be any algorithm for the system with $O(1)$-bit memory and $O(1)$-bit storage. 
We simulate this algorithm by the one-bit memory agent. Let us denote by $\Sim(A)$ the
simulator algorithm, and $\ARmem$ be the memory of the simulator agent. 
In $\Sim(A)$, the storage of each node $v$ consists of two variables $\Nmem$ and $\Amem$, 
and an auxiliary integer variable $\Trans$. Intuitively, $\Nmem$ works as the storage of the simulated algorithm $A$, 
and $\Amem$ is used as the memory of the simulated algorithm. The contents of $\Nmem$ and $\Amem$ are represented by  
fixed-length binary sequences. Let $c$ be the length of $\Amem$.
Our simulation algorithm satisfies the following invariants at the beginning of the one-step simulation: (1)
For any $v \in V$, $b(v).\Trans = 0$ and $b(v).\Nmem =b^{A}(v)$, (2) $b(\Vcur^A).\Amem = \ARmem^A$, and 
(3) $\Portin = \Portin^A$ and $\Vcur = \Vcur^A$.

We refer to the node $\Vcur^A$ as the \emph{departure node} and refer to the neighbor of $\Vcur^A$ the simulated agent 
visits next as the \emph{destination node}. Since the simulator agent has only one-bit memory, it cannot take 
any additional control information. The $j$-th entry of $\Amem$ is referred to as $\Amem[j]$ ($0 \leq j \leq c -1 $).
Hence the agent needs to know the current context only from the information of the 
storage at the departure node and the destination node. The detailed description of the simulator 
agent at a node $u$ is given as follows:
\begin{itemize}
\item If $b(u).\Trans > 0$ when the agent visits $u$ but $b(u).\Trans$ is reset to zero in the local computation 
(where the local computation up to the reset belongs to the past one-step simulation), the simulator agent starts 
the simulation of $A$ regarding $u$ as the departure node. Using the space $b(u).\Amem$ as 
the first tape of $M(A)$ and $b(u).\Nmem$ as the second tape of $M(A)$,
it simulates the run of $M(A)$. The information of $\Portin^A$ is available due to the third invariant above.
When $M(A)$ terminates, the fifth tape correctly stores the destination node. Then the simulator 
agent fetches $b(u).\Amem[b(u).\Trans]$ to $\mathsf{a}$, decrements $b(u).\Trans$, 
and moves to the computed destination node.
\item If $0 \leq b(u).\Trans(u) < c - 1$ holds when the simulator agent visits $u$, the current situation is recognized 
as in transferring the information of $\Amem$ and at the destination node. Hence the simulator agent sets 
$b(u).\Amem[b(u).\Trans] = \ARmem$, increments $b(u).\Trans$, and goes back to $\Portin$.
\item If $- (c - 1) < b(u).\Trans(u) < 0$ holds when the simulator agent visits $u$, the current situation is recognized as 
in transferring the information of $\Amem$ and at the departure node. Hence the simulator agent fetches $b(u).\Amem[- (b(u).\Trans)]$ to $\mathsf{a}$, decrements $b(u).\Trans$, and goes back to $\Portin$.
\item If $b(u).\Trans(u) = - (c - 1)$ holds when the simulator agent visits $u$, the simulator agent finds that transferring $\Amem$ finishes and the current location is the departure node. Hence 
it resets $b(u).\Trans$ to zero and moves to $\Portin$.
\item If $\Trans(u) = c - 1$ holds when the simulator agent visits $u$, the simulator agent finds
that transferring $\Amem$ finishes and the current location is the destination node.
Hence by resetting $b(u).\Trans$ to zero, the current one-step simulation finishes and then proceeds to the next one-step
simulation. 
\end{itemize}
It is easy to check that all the invariants are satisfied at the end of the one-step simulation.
Also, the correctness of the simulation is easily deduced from the invariants. Combining the simulation algorithms of Section~\ref{sec:sim-constant-bit} and \ref{sec:sim-one-bit}, we obtain Theorem~\ref{thm:main}.

\section*{Acknowledgement}
The authors would like to thank Shantanu Das for his suggestion on the simulation by one-bit agent. 

\newcommand{\etalchar}[1]{$^{#1}$}

\section{Correctness of the R-Path Structure}

\subsection{Pseudocode of {\Modifymove}}

To present the algorithm of {\Modifymove}, we introduce an auxiliary procedure 
called \Mark$(I)$. It updates the whiteboard $b(\pi^{-1}_{\Vcur}(\Portin))$ following the instruction $I$. The 
actual behavior inside the procedure is that the agent first goes to $\pi^{-1}_{\Vcur}(\Portin)$, updates 
the contents of $b(\pi^{-1}_{\Vcur}(\Portin))$ following $I$, and goes back to {\Vcur}. Note that the value of 
$\Portin$ and $v_\mt{cur}$ does not change before and after the run of the procedure. The pseudocode of the procedure
{\Modifymove} is presented in Algorithm \ref{alg5}.

\begin{algorithm*}[t] 
\caption{\Modifymove($x$)}   
\label{alg5}
\begin{algorithmic}[1]
\State {\Mark(\Color $\ot$ yellow)} \Comment{Color the new target node.}
\State {$\Portin \ot $ \Findfirst(HeadAscend, {\inPath} $=$ \True $ \land$ {\Color} $=$ \mrm{yellow})}
\If {$\Portin \neq -1$} \Comment{ Find the yellow node on the path.}
  \State {$b(\Vcur)$.(\msf{target}, \msf{inPath}) $\ot$ (\False, \False) }
  \State {\Mark({\target} $\ot$ \True; {\Color} $\ot$ \mrm{white})}
  \State {Move to {\Portin} and Halt}
\Else
  \State {$q \ot $ Find; \Movetop()}
\EndIf
\While{True}
  \If{$q = $ Find} 
    \State {$\Portin \ot $\Findfirst(HeadAscend, {\Color} $=$ \mrm{yellow})}
    \If {$\Portin \neq -1$} \Comment{Finds the yellow(target) node from the neighbors.}
      \State {$q \ot $ Delete}
      \State {$b(\Vcur)$.{\Color} $\ot$ \mrm{red}} \Comment {$\Vcur$ is the predecessor of the new target on R-path.}
      \If {$b(\Vcur)$.{\direction} $=$ \UlargerD} 
        \State {$\Portin \ot $\Findfirst(MiddleAscend, {\inPath} $=$ \True)}
        \If {$\Portin = -1$} \Comment{backward $<$ new target}
          \State {$b(\Vcur)$.{\Color} $\ot$ \mrm{blue}} 
        \EndIf
      \Else
          \State {$\Portin \ot $\Findfirst(MiddleDescend, {\inPath} $=$ \True)}
          \If {$\Portin \neq -1$} \Comment{backward $<$ (forward $<$) new target}
            \State {$b(\Vcur)$.{\Color} $\ot$ \mrm{blue}} 
        \EndIf
      \EndIf
    \EndIf      
    \State {\Moveonehopdown()}
  \EndIf
  \If{$q=$ Delete} 
    \If {$b(\Vcur)$.{\target} $=$ \True}
      \State {$b(\Vcur)$.(\msf{target}, \msf{inPath}) $\ot$ (\False, \False)} \Comment{Removing the old target}
      \State {$\Portin \ot $ \Findfirst(HeadAscend, {\Color} $=$ yellow)}
      \State {Move to \Portin} \Comment {Move to the new target node.}
      \State {$b(\Vcur)$.(\msf{color}, \msf{inPath}, \msf{target}) $\ot$ (\mrm{white}, \True, \True)} \Comment{Set up the new target}
      \State {\Portin $\ot$ \Findfirst(HeadAscend, {\Color} $=$ \mrm{red})}
      \If {$\Portin = -1$} \Comment{Branching node is blue}
        \State {$\Portin \ot$ \Findfirst(HeadAscend, {\Color} $=$ \mrm{blue})}
        \State {\Mark({\direction} $\ot \UsmallerD$)}
      \Else  \Comment {Branching node is red}
        \State {\Mark({\direction} $\ot \UlargerD$)} 
      \EndIf
        \State {\Mark ({{\Color} $\ot$ \mrm{white}})} \Comment {Clean up the predecessor in the path}
        \State {\Mark ({\inPath $\ot$ \True)}}
      \State {Halt}
    \Else
      \State {\Moveonehopdown()} \Comment {Traversal of the expired path}
      \State {\Mark({\inPath} $\ot$ \False)}
    \EndIf
  \EndIf
\EndWhile
\end{algorithmic}
\end{algorithm*}

\subsection{Correctness Proof}

It is obvious that the success of the operations {\Movetop} and {\Moveonehopdown} relies 
on whether the procedure {\Modifymove} correctly manages the contents of variables 
{\inPath} and {\direction}. 
We first present the correctness criteria of the R-path.

\begin{definition} \label{dfn:rpath}
We say that an instance $X$ of the R-path is \emph{consistent} with respect to 
the source node $s$ and the target node $t$ if the following conditions
are satisfied.
\begin{enumerate}
\item $b(t).{\target} = \True$..
\item Each connected component of the subgraph $G'$ of $G$ induced by the nodes satisfying $\inPath = \True$ 
is a path graph.
\item Let $P' = v_1, v_2, \cdots, v_h$ be any connected component $G'$ such that
$v_1 = x$ or $v_h =  y$ holds. For any $v_{i+1} (1 \leq i \leq h - 1)$, if $\pi_{v_{i+1}}(v_i)$ is grater than $\pi_{v_{i+1}}(v_{i+2})$, then $b(v_{i+1}) = \UlargerD$ holds,
and $b(v_{i+1}) = \UsmallerD$ holds otherwise.
\end{enumerate}
In addition, if the number of connected components in $G'$ is one, we say that 
the R-path data structure is \emph{strongly consistent}.
\end{definition}
When an instance $X$ is consistent, we call any connected subgraph managed by $X$ 
(not necessarily a connected component) a \emph{fragment}. 
By definition, each fragment is a path graph, and has backward and forward directions determined 
by variable \textsf{direction}. A fragment with top node $x$ and bottom node $y$ is 
called an $x$-$y$ fragment. From the definition, if $X$ is strongly consistent, 
its (unique) connected component is the $s$-$t$ fragment. For a node $z$ in an 
$x$-$y$ fragment, the backward and lower neighbors of $z$ in the fragment are denoted 
by $\Pred(z)$ or $\Succ(z)$.
The correctness of R-path is captured by the following two lemmas:

\begin{lemma}
Assume that the R-path is consistent and $\Vcur$ is contained in a $x$-$y$ fragment.
Then after running {\Moveonehopdown} or {\Movetop}, the agent stays at $\Succ(v)$
or $x$ respectively.
\end{lemma}

\begin{proof}
As we discussed above, as long as the R-path is consistent, the agent correctly
finds backward or forward in-path neighbors. Thus the lemme obviously holds. 
\end{proof}

\begin{lemma}
At the invocation of {\Moveonehopdown} and {\Movetop} in any execution of {\Modifymove}, 
the R-path is consistent (and thus those operations work correctly). In addition,
at the end of executing {\Modifymove}, the R-path is strongly consistent.
\end{lemma}
\begin{proof}
In Find phases, the agent does not modify {\inPath} or {\direction} 
at any node. Thus the lemma obviously holds. In Delete phases, letting $u$ be the 
branching node, the $s$-$u$ fragment is preserved during the phase, and the agent
gradually deletes the information on the prefix of the $u$-$t$ fragment. More precisely,
if the agent stays at node $u'$ in the Delete phase, $u'$-$t$ fragment
is preserved (and thus consistent). That is, {\Moveonehopdown} performed 
at $u'$ correctly works. It implies that  $u$-$t$ fragment is correctly deleted,
and the agent eventually reaches $t'$. At $t'$, $b(t').\inPath$ becomes true, and
thus the remaining issue is to show that the subset $P'$ consisting all the nodes 
in the $s$-$u$ fragment and $t'$ satisfies the strong consistency. It is obvious that
$P'$ induces one connected subgraph. In addition, except for $u$, no node 
in $P'$ has $t'$ as its neighbor. That is, $P'$ induces a path graph. On the information stored in variable {\direction}, 
no node except for $u$ takes any update, and thus the correctness is preserved. 
The variable $b(u).\direction$ is consistently updated in
the run of {\Modifymove}, which is verified by checking all the cases of Figure~\ref{fig:branching}. The lemma is proved. 
\end{proof}

\section{Omitted Details of Simulation Algorithms from $O(n)$ Bits to $O(1)$ Bits}

\paragraph*{Initialization Phase}
Only the non-trivial point deferred in the main body of the manuscript is the details of 
initializing the third and fourth tapes. We focus on the implementation of initializing 
the third tape, i.e., writing $\Delta(\Vcur^A)$ symbols of 1 to the third tape. First, 
the simulator agent marks the node $\Vcur^A$, and then visits all the neighbors of $\Vcur^A$ 
in the ascending order of their port numbers as ${\Findfirst}(\mathrm{HeadAscend}, \cdot)$. 
During this process, the location of the simulator agent is always tracked by the R-path instance 
$X(\Vcur)$. When the simulator agent visits any neighbor $v$ of $\Vcur^A$, there exists 
a traversable path from $v$ to $t_3$ managed by $X(v)$ and $X(t_3)$ Using that path, the simulator 
agent moves to $t_3$ and appends symbol 1 to the tail of the third tape. Then it 
moves the position of $t_3$ by invoking {\Modifysucc} for $X(t_3)$, and returns to $v$. 
To go back to $\Vcur^A$, the simulator agent needs to identify the marked neighbor, which is 
done by {\Findfirst}. After repeating this process until the simulator agent visits all 
neighbors of $\Vcur^A$, the third tape contains $\Delta(\Vcur^A)$ symbols of 1. The 
initialization of the fourth tape is almost the same as that of the third tape. 
Only the difference is to iterate appending symbol 1 until the simulator agent finds the neighbor 
with $\mathit{past} = \True$, instead of visiting all neighbors.

\paragraph*{Agent Movement Phase}
Let $p$ be the number of the port to which the simulated agent moves, i.e., the port number 
written in the fifth tape in the local computation phase. Similarly to the initialization 
phase, the node $\Vcur^A$ is marked at the beginning of the phase. We refer to the target node 
of $X(\Vcur^A)$ as $v$. To change $v$ to $\pi_{\Vcur^A}(p)$, the simulator agent
first places it on $\pi_{\Vcur^A}(0)$. Assume that $v = \pi_{\Vcur^A}(j)$ for some $0 \leq j 
\leq \Delta(\Vcur^A) - 1$. Then the simulator agent moves to $t_5$ using the path managed by
$X(v)$ and $X(t_5)$, and checks if the symbol 1 is still contained in the fifth tape or not.
If the fifth tape does not contain 1, then $j = p$ holds and thus the movement phase finishes.
Otherwise, after deleting one symbol 1 in the fifth tape, the simulator agent goes back to 
$v = \pi_{\Vcur^A}(j)$. Then it changes the location of $v$ to $\pi_{\Vcur^A}(j + 1)$,
which is realized by invoking {\Modifymove} twice: The agent first finds the marked neighbor of 
$v$ (i.e., $\Vcur^A$), and moves $v$ to there by invoking {\Modifymove} for $X(v)$. 
Since the simulator agent moves from $\pi_{\Vcur^A}(j)$ to $\Vcur^A$ at the last step of 
{\Modifymove}, $\Portin$ stores the value $j$ immediately after the first invocation of {\Modifymove}. 
By going out from $\Vcur^A$ through the port $\Portin + 1$ and coming back to 
$\Vcur^A$, $\Portin$ is updated by $j+1$. The second invocation of {\Modifymove} for $X(v)$ 
takes $v$ to $\pi_{\Vcur^A}(j + 1)$. Finally, the phase finishes with unmarking $\Vcur^A$
and setting $b(\Vcur^A).\mathsf{past} = \True$.

\section{Proof of Theorem~\ref{thm:impossible}}
\label{sec:impossible}
In this section, we prove Theorem~\ref{thm:impossible} by showing that the following 
task cannot be solved only by the agent with zero-bit memory.
\begin{problem}[Task \Parity]
The task {\Parity} consists of all the graphs with a odd number of vertices.
\end{problem}
It is easy to show that $\Parity \in \Pagt(O(n),\Theta(1)) = \Pagt(1, \Theta(1))$ holds: The
agent counts up the modulo 2 of $n$ by running \textsf{DLDFS}. However, this task cannot be 
solved by the zero-bit agent.
\begin{lemma}
\label{lma:impossible}
$\Parity \notin \Pagt(0,\Theta(g(1)))$.
\end{lemma}

\begin{proof}
Suppose for contradiction that an algorithm $A$ solves {\Parity} using zero-bit memory and $O(1)$-bit storage.
Due to the impossibility result by Cohen et al.~\cite{CFIKP08} there exists no graph exploration algorithm
using zero-bit memory and $O(1)$-bit storage. Hence there exists the graph $G = (V, E, \Pi)$ and the initial location 
$v \in V$ such that $A$ decides if $G \in {\Parity}$ or not without vising all the nodes in $G$. Let $u \neq v$ 
be the node not visited by the agent in the execution of $A$, and constructs the graph $G' = (V', E', \Pi')$ 
obtained from $G$ by adding one node $u'$ and the edge $(u, u')$. Since the agent does not visit $u$, it cannot 
distinguish $G$ and $G'$ and thus the decision for those graphs are the same. However, the parity of $|V|$ and 
$|V'|$ are different. It is a contradiction.
\end{proof}

This lemma obviously deduces Theorem~\ref{thm:impossible}.

\end{document}